\documentclass[10pt]{article}

\usepackage{amsthm,amsmath,amsfonts,braket,algorithm,graphicx, braket}
\usepackage{fullpage,authblk}

\theoremstyle{definition} 
\theoremstyle{definition} 
\newtheorem {theorem} {Theorem}

\newtheorem {lemma} {Lemma}

\newcommand{\kb}[1]{\ket{#1}\bra{#1}}

\newcommand{\wideq}{\widetilde{q}_A,\widetilde{q}_B}

\newcommand{\al}{\mathcal{A}}

\newcommand{\trd}[1]{\left|\left| #1 \right| \right|}

\newcommand{\st}{\text{ } : \text{ }}

\newcommand{\Hmin}{H_\infty}

\newcommand{\Hextd}{\bar{H}}

\newcommand{\hd}{\Delta_H}

\newcommand{\dc}{\sim_\delta}

\newcommand{\leak}{\texttt{leak}_{\text{EC}}}

\floatname{algorithm}{Protocol}

\begin{document}
\title{Security of a High Dimensional Two-Way Quantum Key Distribution Protocol}

\date{}
\author[1]{Walter O. Krawec\footnote{Email: \texttt{walter.krawec@uconn.edu}}}
\affil[1]{\small{Department of Computer Science and Engineering}\\\small{University of Connecticut}\\\small{Storrs, CT 06269 USA}}
	
\maketitle
	
\begin{abstract}
Two-way quantum key distribution protocols utilize bi-directional quantum communication to establish a shared secret key.  Due to the increased attack surface, security analyses remain challenging.  Here we investigate a high-dimensional variant of the Ping Pong protocol and perform an information theoretic security analysis in the finite-key setting.  Our methods may be broadly applicable to other QKD protocols, especially those relying on two-way channels.  Along the way, we show some fascinating benefits to high-dimensional quantum states applied to two-way quantum communication.
\end{abstract}

\section{Introduction}
Quantum Key Distribution (QKD) allows two parties to establish a shared secret key, secure against computationally unbounded adversaries.  This is in contrast to classical key distribution where computational assumptions are always required for security.  Typically, QKD protocols utilize a one-way quantum channel, allowing Alice to send quantum resources to Bob.  However, two-way quantum channels, allowing for bi-directional quantum communication between Alice and Bob, are also a possibility and have several exciting benefits, at least in theory, over one-way protocols.  For instance, deterministic key distribution is possible along with secure direct communication \cite{QKD-PP,QKD-SDC,QKD-LM05,PhysRevA.69.052319} and protocols with devices restrictions, such as the so-called ``semi-quantum'' model of cryptography \cite{SQKD-first,SQKD-second,SQKD-survey}.  Experimental implementations of such protocols are also possible \cite{TwoWay-Exp1,TwoWay-Exp2,massa2019experimental,gurevich2013experimental}.  They are also interesting in their own way as they show how alternative methods, such as super dense coding (SDC) \cite{SDC} can be used for QKD (as in the Ping Pong protocol \cite{QKD-PP} or its extension described in \cite{QKD-SDC}).   Certain fiber implementations also may benefit from two-way quantum communication as polarization drift can be compensated for in the return channel \cite{lucamarini2014quantum,two-way-qkd-security}.

Naturally, there are also disadvantages to such two-way systems, in particular when considering photon loss over fiber (as a signal must now travel twice as long) or side-channel attacks against devices.  Furthermore, the two-way channel introduces a second opportunity for Eve to attack each quantum signal making security analyses very difficult in general.  Despite this, the study of two-way QKD protocols is still of importance for a variety of reasons: first, they exhibit, as mentioned, several advantages over one-way protocols and, second, their study may lead to breakthroughs in other fields of quantum cryptography and quantum information science (e.g., in new proof techniques, protocol design methods, or countermeasures to noisy channels).  See \cite{qkd-survey-pirandola,qkd-survey-scarani,amer2021introduction} for general surveys of QKD, both theory and practice.

QKD protocols relying on a two-way channel are arguably not as well understood as their one-way counterparts, especially in the practical finite-key setting.  One of the most powerful methods of proving QKD security is to reduce a protocol to an equivalent entanglement based version and then utilize entropic uncertainty \cite{maassen1988generalized,berta2010uncertainty,survey,survey-2,survey-3,tomamichel2017largely}.  However, these methods typically cannot be applied directly to two-way protocols due to the ability for an adversary to interact twice with a quantum signal (with the second attack perhaps depending on the first attack and the receiving party's actions).  Typical reduction based proofs to entanglement based protocols simply don't work for the most part.

Rather interestingly, in \cite{two-way-qkd-security}, a method for reducing certain types of two-way protocols to equivalent entanglement based versions was shown (thus allowing for the use of entropic uncertainty to derive key-rates and prove security).  However, that result only applied to protocols which held a certain symmetry property which not all do.  In such a protocol, Eve would prepare a tripartite state, sending part to Alice, part to Bob, and keeping part for herself, concluding the quantum communication stage of the protocol.  However, this symmetry property does not always hold, especially for protocols where users are restricted in some way (e.g., they are semi-quantum \cite{SQKD-first,SQKD-second,SQKD-survey} or have other measurement/preparation limitations as with the protocol we consider in this work).  Thus, it is important to develop alternative techniques to handle such scenarios.

In this paper, we analyze a two-way protocol based on the Ping Pong protocol introduced in \cite{QKD-PP}.  At a high level, this protocol involves Alice creating a Bell pair $\ket{\psi} = \frac{1}{\sqrt{2}}(\ket{00} + \ket{11})$ and sending one particle to Bob while keeping the other private.  Bob will then encode his key-bit on the qubit by either performing an identity operator (if his key-bit is zero) or a phase flip $\sigma_Z$ (if his key-bit is one).  On return, Alice will measure in the Bell basis to determine the key-bit.  Obviously if this were the entirety of the protocol it would be completely insecure; thus Alice and Bob must perform suitable checks by measuring in alternative bases to guarantee security.  Interestingly, an extended version of this, using the full super-dense-coding (SDC) protocol was discussed in \cite{QKD-SDC}.

In this work, we consider a high-dimensional variant of the Ping Pong protocol introduced in \cite{QKD-HD-PP}.  However, we consider a restricted version of this protocol.  The protocol we consider in this work is more similar to the original Ping Pong protocol in that it does not fully utilize all possible SDC encodings however, in doing so, it also significantly simplifies the measurement capabilities needed by users (especially as the dimension increases).  This makes it potentially easier to implement in practice, and also introduces some interesting theoretical problems and insights.  For instance, we show later that the noise tolerance of the Ping Pong protocol and the more complicated SDC protocol in \cite{QKD-SDC} are identical in the qubit case.  The noise tolerance of both also match that of BB84 in the qubit case. However as the dimension increases, high-dimensional BB84 (HD-BB84 \cite{bechmann2000quantum}) always outperforms the two-way protocol we analyze here (in stark contrast to the qubit case).

Since the protocol we analyze does not contain the necessary symmetry property for results in \cite{two-way-qkd-security} to be applied, we are also forced to use alternative methods to prove security.  We develop an alternative approach, based on the framework of quantum sampling by Bouman and Fehr \cite{bouman2010sampling} and our recent proof methods for sampling based entropic uncertainty \cite{krawec2019quantum}, to compute the quantum min entropy needed to evaluate finite key-rates.  We also consider the case of imperfect measurement devices and lossy channels.  Our proof of security may be highly beneficial to other QKD protocols.  To our knowledge, we are the first to derive a rigorous finite key analysis for this high-dimensional two-way protocol (the qubit case was analyzed in \cite{shaari2015finite}).  Note that our proof also covers the qubit case, thus giving an alternative proof of security for the original qubit-based Ping Pong protocol.

We make several contributions in this work.  We develop a proof of security in the finite-key setting for general attacks, in both the ideal device scenario and the non-ideal device scenario (including lossy channels and imperfect detectors).  Our proof method may be applicable to other QKD protocols (both one way and two-way) where standard tools cannot be directly applied, thus giving new mathematical tools for other researchers to utilize when proving other protocols secure.  Finally, we perform a rigorous analysis of our resulting key-rates for various dimensions, showing some fascinating properties of the protocol as the dimension of the quantum signal increases.  High-dimensional quantum states have been shown to be very beneficial to various QKD protocols \cite{bechmann2000quantum,chau2005unconditionally,sheridan2010security,sasaki2014practical,chau2015quantum,vlachou2018quantum,cerf2002security,nikolopoulos2005security,iqbal2020high,nikolopoulos2006error,yin2018improved,doda2021quantum} (see \cite{HD-qkd-survey} for a recent survey on high-dimensional quantum communication); we show in this work more evidence that they can also benefit two-way protocols in the finite key setting.

\section{Notation and Definitions}

We begin by introducing some notation and preliminary concepts which we will use throughout this paper.  We use $\al_d$ to denote an alphabet of $d$ characters which has a distinguished $0$ character.  Without loss of generality, we simply assume $\al_d = \{0, 1, \cdots, d-1\}$.  Given a word $q \in \al_d^n$ and a subset $t \subset \{1, \cdots, n\}$, we write $q_t$ to mean the substring of $q$ indexed by $t$, namely: $q_t = q_{t_1}q_{t_2}\cdots q_{t_{|t|}}$.  We write $q_{-t}$ to mean the substring indexed by the complement of $t$.  We use $\mathcal{H}_d$ to denote a $d$-dimensional Hilbert space.

We use $w(q)$ to mean the relative Hamming weight of $q$ defined to be: $w(q) = \frac{|\{i \st q_i \ne 0\}|}{|q|}$.  Given $x,y \in \al_d^n$, we use $\hd(x,y)$ to be the Hamming distance between words $x$ and $y$, namely $\hd(x,y) = \frac{|\{i \st x_i \ne y_i\}|}{n}$.  Finally, let $\delta \ge 0$, then, given two real numbers $x, y$, we write:
\[
x\dc y \iff |x - y| \le \delta.
\]

Given a density operator $\rho_{AB}$ acting on some Hilbert space $\mathcal{H}_A\otimes\mathcal{H}_B$, we write $\rho_B$ to mean the state resulting from tracing out the $A$ system, namely $\rho_{B} = tr_A\rho_{AB}$.  Similarly for the other system.  The conditional quantum min entropy \cite{renner2008security} is defined to be:
\begin{equation}
\Hmin(A|B)_\rho = \sup_{\sigma_B}\max\{\lambda\in\mathbb{R} \st 2^{-\lambda}I_A\otimes \sigma_B - \rho_{AB} \ge 0\},
\end{equation}
where $X \ge 0$ implies $X$ is positive semi-definite.  When the $B$ system is trivial, it is not difficult to see that:
\[
\Hmin(A)_\rho = -\log_2\max\{\lambda \st \lambda \text{ is an eigenvalue of $\rho_A$}\}.
\]
The \emph{smooth conditional min entropy} is defined to be \cite{renner2008security}:
\begin{equation}
\Hmin^\epsilon(A|B)_\rho = \sup_{\sigma_{AB}\in\Gamma_\epsilon(\rho_{AB})}\Hmin(A|B)_\sigma
\end{equation}
where $\Gamma_\epsilon(\rho) = \{\sigma \st \trd{\rho-\sigma} \le \epsilon\}$ and where $\trd{X}$ is the trace distance of operator $X$.

Quantum min entropy is a highly useful quantity to measure in quantum cryptography as it is directly related to the number of uniform independent random bits that may be extracted from a classical-quantum (cq)-state $\rho_{AE}$.  In particular, as shown in \cite{renner2008security}, given cq-state $\rho_{AE}$, following a privacy amplification process, consisting of hashing the classical $A$ register using a two-universal hash function with output size $\ell$-bits resulting in final cq-state $\sigma_{KE}$, it holds that:
\begin{equation}\label{eq:PA}
  \trd{\sigma_{KE} - \frac{I_K}{2^\ell}\otimes\sigma_{E}} \le 2^{\frac{1}{2}(\Hmin(A|E)_\rho - \ell)} + 2\epsilon.
\end{equation}
Thus, to determine the final size of the secret key following the use of any QKD protocol, one requires a bound on the quantum min-entropy of the state $\rho_{AE}$ produced before privacy amplification.  Note that, in a QKD protocol, error correction leaks additional information which must also be taken into account.

An important min entropy expression we will need involves quantum-quantum-classical (qqc) states of the form $\rho_{ABC} = \sum_{c\in\mathcal{C}} p_c \rho_{AB}^{(c)}\otimes\kb{c}$ for some finite set $\mathcal{C}$.  Given such a state, it is straight forward to show from the definition of min entropy that:
\begin{equation}
\Hmin(A|BC)_\rho \ge \min_{c\in\mathcal{C}}\Hmin(A|B)_{\rho^{(c)}}.
\end{equation}
From this, it is also easy to show that, given $\rho_{AB} = \sum_{c\in\mathcal{C}}p_c\rho_{AB}^{(c)}$, it holds that:
\begin{equation}\label{eq:mixed-entropy}
  \Hmin(A|B)_\rho \ge \min_{c\in\mathcal{C}}\Hmin(A|B)_{\rho^{(c)}}.
\end{equation}
The above follows by appending an auxiliary system $C$ that is classical and noting that $\Hmin(A|B) \ge \Hmin(A|BC)$.

The \emph{Shannon entropy} of a random variable $X$ is denoted $H(X)$.  If $X$ is a two outcome random variable with the probability of one outcome being $p$, then we use $H(p)$ to mean the binary Shannon entropy, namely $H(p) = -\log_2p - (1-p)\log_2(1-p)$. The $d$-ary entropy function is denoted $H_d(x)$ and defined to be:
\begin{equation}
H_d(x) = x\log_d(d-1) - x\log_dx - (1-x)\log_d(1-x).
\end{equation}
Observe that $H_2(p) = H(p)$.

Later, we will prove security of the protocol under investigation by considering an ``ideal'' state that is $\epsilon$ close to a real state, on average over a separate, classical, subsystem.  The following lemma will allow us to argue about the entropy in the real state assuming one can bound the entropy of the ideal state under all possible classical outcomes.  This lemma is a generalization of something we proved in \cite{krawec2020new}; here we generalize the result to a broader range of applications:
\begin{lemma}\label{lemma:entropy}
  Let $\rho, \sigma,$ and $\tau$ be three quantum states such that $\frac{1}{2}\trd{\rho-\sigma} \le \epsilon$.  Let $\mathcal{F}$ be a CPTP map such that:
\[
  \mathcal{F}\left(\tau\otimes \rho\right) = \sum_xp_x\kb{x}_X\otimes\rho_{AE}^{(x)}
\]
and
\[
  \mathcal{F}\left(\tau\otimes \sigma\right) = \sum_xq_x\kb{x}_X\otimes\sigma_{AE}^{(x)}
\]
Then, it holds that:
\begin{equation}
  Pr\left(\Hmin^{4\epsilon+2\epsilon^{1/3}}(A|E)_{\rho^{(x)}} \ge \Hmin(A|E)_{\sigma^{(x)}}\right) \ge 1-2\epsilon^{1/3},
\end{equation}
where the probability is over the random variable $X$.  Note that the $\tau$ system may be trivial if not needed; however it may be used to represent, for instance, additional seed randomness needed by $\mathcal{F}$.
\end{lemma}
\begin{proof}
Since CPTP maps do not increase trace distance, we have:
\begin{align*}
2\epsilon & \ge \trd{\rho-\sigma} = \trd{\tau\otimes(\rho-\sigma)} \ge \trd{\sum_x p_x\kb{x} \otimes\rho_{AE}^{(x)} - \sum_x q_x\kb{x}\otimes\sigma_{AE}^{(x)}}\\
&\ge \trd{\sum_x p_x\kb{x}\otimes\left(\rho_{AE}^{(x)} - \sigma_{AE}^{(x)}\right) - \sum_x\left(q_x - p_x\right)\kb{x}\otimes\sigma_{AE}^{(x)}}\\
&\ge \trd{\sum_x p_x\kb{x}\otimes\left(\rho_{AE}^{(x)} - \sigma_{AE}^{(x)}\right)} - \sum_x|q_x - p_x|,
\end{align*}
where, for the last inequality, we used the reverse triangle inequality.  Note that $\frac{1}{2}\trd{\rho-\sigma}\le\epsilon$ implies that $\frac{1}{2}\sum_x|q_x-p_x|\le\epsilon$ and so we have:
\begin{equation}
\frac{1}{2}\trd{\sum_x p_x\kb{x}\otimes\left(\rho_{AE}^{(x)} - \sigma_{AE}^{(x)}\right)} \le 2\epsilon.
\end{equation}
Let $\Delta_x = \frac{1}{2}\trd{\rho_{AE}^{(x)} - \sigma_{AE}^{(x)}}$.  We treat this as a random variable over the choice of $x$.  From the above, it is clear that the expected value $\mu = \mathbb{E}(\Delta_x) \le 2\epsilon$.  Furthermore, the variance may also be bounded by $V^2 \le 2\epsilon$ (this is due to the fact that $\Delta_x \le 1$ for all $x$).  Using Chebyshev's inequality it holds that:
\begin{equation}
  Pr\left(|\Delta_x - \mu| \le \epsilon^{1/3}\right) \ge 1 - 2\epsilon^{1/3},
\end{equation}
thus, except with probability $2\epsilon^{1/3}$, it holds that:
\[
\Delta_x \le \epsilon^{1/3} + \mu \le \epsilon^{1/3} + 2\epsilon.
\]
Switching to smooth min entropy completes the proof.
\end{proof}

\subsection{Quantum Sampling}\label{sec:sample}

Our proof technique will make use of a quantum sampling methodology introduced by Bouman and Fehr in \cite{bouman2010sampling}.  This method has been shown to have many applications including the proof of BB84 \cite{bouman2010sampling}; the proof of security of a high-dimensional BB84 protocol \cite{krawec-sample-qkd}; and also the development of novel, so-called \emph{sampling-based entropic uncertainty relations} \cite{krawec2019quantum,krawec2020new}.  We review the terminology and concepts from \cite{bouman2010sampling} in this section, making some generalizations.

Fix an alphabet $\al_d$ and a number $N > 1$.  We define a classical sampling strategy as a triple $(P_T, g, r)$ where $P_T$ is a probability distribution over all subsets of $\{1, \cdots, N\}$ and $g,r:\al_d^*\rightarrow \mathbb{R}$.  The function $g$ is a ``guessing function'' while the function $r$ is a ``target function.''  Then, given some word $q \in \al_d^N$, the strategy consists of the following process: first, sample a subset $t$ according to $P_T$.  Next, given the observation $q_t$ (i.e., given the actual value of $q$ on subset $t$), evaluate $g(q_t)$ to produce a ``guess'' as to the value of $r(q_{-t})$.  The guess (using only the observed portion) should be close to the actual target evaluation (of the \emph{unobserved} portion) with high probability.  In particular, for a good sampling strategy, it should be that, with high probability over the subset choice, one has $g(q_t) \dc r(q_{-t})$.  As an example, one may take $g(x) = r(x) = w(x)$.  In this case, the sampling strategy attempts to guess at the relative Hamming weight in the unobserved portion of some string and the guess is simply the relative Hamming weight of the observed portion.

More formally, fix $\delta > 0$ and some subset $t$.  Then we define the set of ``good'' words in $\al_d^N$ to be:
\[
\mathcal{G}_{t,\delta} = \{q \in \al_d^N \st   r(i_{-t}) \dc g(i_t)\}.
\]
This set consists of all words in the alphabet such that the guess function, given some observed portion of the word and for a given, fixed, subset $t$, produces a $\delta$-close guess at the target function in the unobserved portion.  Namely, this set consists of all words for which the sampling strategy is guaranteed to work assuming the given subset $t$ was actually chosen.  The \emph{error probability} of the given classical sampling strategy, then, is defined to be:
\begin{equation}
\epsilon_\delta^{cl} = \max_{q \in \al_d^N} Pr\left(q \not\in \mathcal{G}_{t,\delta}\right),
\end{equation}
where the probability is over all subsets $t$ chosen according to $P_T$.  The ``$cl$'' superscript is used to indicate that this is the failure probability of a ``classical'' sampling strategy.  It is easy to see that, given a word $q \in \al_d^N$, the probability that the sampling strategy fails to produce a $\delta$-close guess at the target value for this particular word is at most $\epsilon_\delta^{cl}$.

To translate these notions into a \emph{quantum} sampling strategy, fix an orthonormal basis $\{\ket{0},\cdots,\ket{d-1}\}$ of the $d$-dimensional Hilbert space $\mathcal{H}_d$.  From this, consider $\text{span}(\mathcal{G}_{t,\delta})$ defined to be:
\[
\text{span}(\mathcal{G}_{t,\delta}) = \text{span}\left\{\ket{q} \st q \in \mathcal{G}_{t,\delta}\right\},
\]
where $\ket{q} = \ket{q_1}\otimes\cdots\otimes\ket{q_N}$ (again, where each $\ket{q_i}$ is with respect to the given, but arbitrary, orthonormal basis).  Similar to how $\mathcal{G}_{t,\delta}$ represents the set of classical ``good'' words, a state $\ket{\phi}$ living in the above spanning set will be called an \emph{ideal quantum state}.  It is one where sampling in the given basis is guaranteed to produce a collapsed state that acts in a well behaved and understood manner, if using fixed subset $t$.  In particular, given $\ket{\phi}_{AE} \in \text{span}(\mathcal{G}_{t,\delta})\otimes\mathcal{H}_E$ (where the $E$ system may be arbitrarily entangled with the $A$ portion), if the $A$ portion, indexed by fixed subset $t$, is measured in the given basis resulting in outcome $x$, then it is guaranteed that the unmeasured portion of the state collapses to a superposition of the form:
\[
\ket{\phi(x)} = \sum_{i\in J_x}\alpha_i\ket{i}\otimes\ket{E_i},
\]
where $J_x = \{i\in \al_d^{N-|t|} \st  r(i_{-t}) \dc g(x)\}$.

From this, the main result from \cite{bouman2010sampling} can be stated in the following theorem:

\begin{theorem}\label{thm:sample}
  (From results in \cite{bouman2010sampling}): Let $\delta \ge 0$ and $(P_t, g, r)$ be a classical sampling strategy with failure probability $\epsilon_\delta^{cl}$.  Then, for any $\ket{\psi}_{AE} \in \mathcal{H}_A\otimes\mathcal{H}_E$ (typically called the ``real state''), where the $A$ system is an $N$-fold tensor of $\mathcal{H}_d$, there exist a collection of ideal states $\{\ket{\phi^t}_{AE}\}_t$, indexed over all subsets $t$ with $P_T(t) > 0$, such that:
    \begin{enumerate}
      \item Each ideal state lives in the spanning set of good classical words, namely: $\ket{\phi^t}_{AE} \in \text{span}\left(\mathcal{G}_{t,\delta}\right)\otimes\mathcal{H}_E$
      \item On average over all ideal states, the resulting system is $\sqrt{\epsilon^{cl}_\delta}$ close to the given input state, namely:
        \begin{equation}
          \frac{1}{2}\trd{\sum_tP_T(t)\kb{t}\otimes\left(\kb{\psi}_{AE} - \kb{\phi^t}_{AE}\right)} \le \sqrt{\epsilon^{cl}_\delta}.
        \end{equation}
    \end{enumerate}
\end{theorem}

Note that the above theorem is a slight rewording of the result in \cite{bouman2010sampling}.  For a proof that this version follows from results in \cite{bouman2010sampling}, see our work in \cite{krawec-sample-qkd}.

To use Theorem \ref{thm:sample}, one must first define a sampling strategy which, thus, defines a set $\mathcal{G}_{t,\delta}$.  From this, if one knows the error probability of the classical sampling strategy, Theorem \ref{thm:sample}  may be used to ``promote'' the strategy to a quantum one acting on superposition states.

In our protocol analysis, we will require the following natural two-party sampling strategy.  Given a word $(x,y) \in \al_d^N\times \al_d^N$, with $N = n+m$, a subset $t\subset \{1, \cdots, N\}$ of size $m$ is chosen uniformly at random.  We are then given observation $(x_t, y_t)$.  The target function is the Hamming distance of $x_{-t}$ and $y_{-t}$, namely: $r(x_{-t}, y_{-t}) = \hd(x_{-t}, y_{-t})$. The guess function is the Hamming distance of the observed strings, namely: $g(x_{t}, y_{t}) = \hd(x_t, y_t)$.  In particular, this strategy observes the $t$ portion of both $x$ and $y$ and guesses that the Hamming distance in the unobserved portion is $\delta$-close to the Hamming distance of the observed strings.  Note that the set of good words for this strategy, therefore, is:
\begin{equation}\label{eq:good-words}
  \mathcal{G}_{t,\delta} = \left\{(x,y) \in \al_d^N\times\al_d^N \st \hd(x_t,y_t) \dc \hd(x_{-t},y_{-t})\right\}.
\end{equation}

The following Lemma was proven in \cite{krawec-sample-qkd}:
\begin{lemma}\label{lemma:sample}
  (From \cite{krawec-sample-qkd}): Given the above sampling strategy with $m \le n$ and $\delta > 0$, then:
  \begin{equation}
    \epsilon_{\delta}^{cl} \le 2\exp\left(\frac{-\delta^2 m(n+m)}{m+n+2}\right).
  \end{equation}
\end{lemma}

\section{Protocol}

We consider a variant of the high-dimensional version of the Ping Pong protocol \cite{QKD-HD-PP} which, as discussed in the introduction, utilizes a two-way quantum channel.  We assume the dimension of a single signal state is $d \ge 2$ (when $d=2$ we recover the original Ping Pong protocol).  A single round of the protocol involves Alice preparing an entangled state of the form
\[
\ket{\psi_0} = \frac{1}{\sqrt{d}}\sum_{a=0}^{d-1}\ket{a,a}_{AT}.
\]
Note this is a $d^2$ dimensional state.  She will keep the $A$ register private while sending the $d$-dimensional $T$ register to Bob over the first quantum channel (here, we use ``$T$'' to mean the ``Transit'' register).  Bob, on receipt of the signal, will choose randomly whether this round will be a \emph{Test Round} or a \emph{Key Round}.  If the former, he measures his signal in the computational $Z = \{\ket{0}, \cdots, \ket{d-1}\}$ basis and sends his measurement result to Alice (over the public authenticated channel).  In this case, Alice knows Bob chose this to be a test round, so she, too, will measure her signal in the $Z$ basis and compare with Bob's outcome which should be correlated.  Otherwise, if this is a key round, Bob will choose a random key digit for this round $k \in \al_d$ and encode his choice onto his received signal using the unitary operator:
\begin{equation}
U_k = \sum_{a=0}^{d-1}\exp(2\pi i k a /d)\kb{a}.
\end{equation}
He then returns the signal to Alice.  Note that, the main operational difference between our protocol and that of \cite{QKD-HD-PP} is that Bob has fewer encoding options - this decreases the overall efficiency of the protocol compared to the original, but also decreases user's complexity (interestingly, if Bob were to have more encoding options, he would also have to do more test measurements, besides $Z$, to ensure the fidelity of the state sent from Alice).  Thus, though we limit the transmission rate below what is theoretically possible by creating this $d^2$ state, we also only require Bob to measure in the $Z$ basis as opposed to a number of bases scaling with the dimension of the system as required in the original protocol.  Depending on the encoding system used (e.g., time bin encoding \cite{timebin1} which is a useful encoding method for high-dimensional quantum cryptography \cite{boaron2018simple,islam2017provably,brougham2013security}), this can greatly simplify Bob's device requirements.  It also simplifies Alice's device as she must only measure in two bases.

On receipt of the signal from Bob, Alice will measure both her original $A$ register and the returned $T$ register using POVM $\Lambda = \{\Lambda_0, \cdots, \Lambda_{d-1}\}$ where:
\begin{equation}
\Lambda_k = \sum_a P \left( \sum_j \exp(2\pi i k j / d)\ket{j, j+a}\right),
\end{equation}
where $P(z) = zz^*$.  Note that, in the event of no noise, if Bob chooses a particular value $k$ to encode, it should be that POVM element $\Lambda_k$ will be observed with unit probability.  The protocol then repeats $N$ times, for $N$ sufficiently large.  Each key-round contributes an additional $\log_2 d$ bits to the users' raw keys.  At the conclusion of these $N$ rounds, users will run a standard error correction and privacy amplification protocol to determine their final secret key of size $\ell$ bits.

To analyze the security of this protocol in the finite key setting, we will actually consider a slightly modified version where we allow Eve to prepare the initial signal state.  In particular, Eve will prepare a quantum state $\ket{\psi}_{ATE}$ where the $A$ and $T$ registers each live in a Hilbert space of dimension $d^{N}$, where $d$ is the dimension of a single round's signal state and $N$ is the total number of rounds used in the protocol (a parameter that may be optimized by the users Alice and Bob).  Note that we make no assumptions on this state other than the dimension of the $A$ and $T$ portions.  The $A$ register is given to Alice and the $T$ register is given to Bob.

Alice and Bob next choose a random subset $t$ of size $m \le N/2$ where $m$ is another parameter that users may optimize over (these will be the test rounds).  This process may be done by having Bob choose the subset and sending it to Alice over the authenticated channel; or if $m$ is small enough, it may be produced using a small pre-shared secret key.  We assume the former option but our proof works in the second scenario simply by subtracting $\log_2{N\choose m}$ bits from our final derived secret key $\ell$.  On this test subset, Alice and Bob measure their signals in the $Z$ basis recording the result as $q_A$ and $q_B$.  Both parties send their measurement outcomes to each other over the authenticated channel allowing them to compute $\hd(q_A, q_B)$.  Note that, if Eve were honest and prepared the state $\ket{\psi_0}_{AT}^{\otimes N}\otimes \ket{0}_E$, it should be that $\hd(q_A,q_B) = 0$.  Any $\hd(q_A,q_B) > 0$ will be considered noise and will be factored into our key-rate equation. If the noise is too high, parties will abort.

\begin{figure}
    \centering
    \includegraphics[width=.75\textwidth]{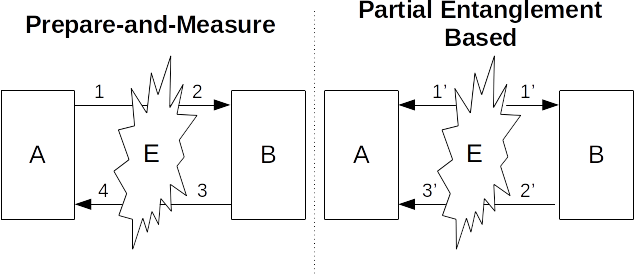}
    \caption{Left: High-level diagram of the two-way QKD protocol.  (1) Alice prepares some quantum system and sends part to Bob; (2) After Eve attacks the signal, Bob receives the quantum state; (3) Bob performs some testing on the state and encodes his raw key on the remaining portion; (4) Finally, Eve attacks a second time and returns the signal to Alice (who attempts to guess the key that Bob encoded while also performing some tests).  Right: Diagram of the partial entanglement based version that we prove secure. (1') Eve prepares some initial state and sends part to Alice, part to Bob and keeps part to herself; (2') Alice and Bob do some tests on their systems and then Bob encodes his key onto the remaining part, sending it back to Alice; (3') Eve attacks a second time and returns a signal to Alice.  Clearly security of the partial entanglement based version (Right) will imply security of the actual protocol under consideration (Left).  Note that this is not a full entanglement based protocol as Bob must return a quantum signal to Alice and Eve is allowed to attack a second time.}
    \label{fig:diagram}
\end{figure}

After this test, Bob will choose a random key $K = k_1\cdots k_n \in \al_d^n$, where $n = N-m$.  He then applies the unitary operator $\mathcal{U}_K = U_{k_1}\otimes\cdots\otimes U_{k_n}$ to the remaining, unmeasured, portion of his signal received from Eve.  After this, he sends his quantum state to Alice who will subsequently measure each state using POVM $\Lambda$ producing her raw-key $\widetilde{K}$ which, in the noiseless case, should match exactly Bob's key encoding of $K$.  Parties then run error correction and privacy amplification as before.  It is not difficult to see that security of this protocol, where Eve prepares the initial state, will imply security of the desired protocol where Alice prepares the state.  Note that this is not a full entanglement based protocol as it still involves two-way quantum communication and Eve still has two opportunities to attack.  Nonetheless, we will show a method to prove general security in this scenario.  See Figure \ref{fig:diagram} for a diagram depicting the two scenarios.

\section{Security Analysis}\label{sec:security1}

We now analyze the security of the partial entanglement based protocol in the finite key setting.  Note that the sampling used in the protocol is exactly the strategy discussed in Section \ref{sec:sample} and analyzed in Lemma \ref{lemma:sample}.  Using this, we have $\mathcal{G}_{t,\delta}$ as defined in Equation \ref{eq:good-words}.

Let $\epsilon > 0$ be given (it will, as we later show, directly relate to the security of the secret key and failure probability of the protocol).  Also, let $m$ and $n$ be fixed, with $N = n+m$.  Here, $N$ will be the number of signals sent while $m$ is the subset size chosen for testing.  Then, we set:
\begin{equation}\label{eq:delta}
  \delta = \sqrt{\frac{(m+n+2)\ln(2/\epsilon^2)}{m(m+n)}}.
\end{equation}
Note that, from Lemma \ref{lemma:sample}, this implies $\sqrt{\epsilon_\delta^{cl}} = \epsilon$.

Let $\ket{\psi}_{ATE}$ be the initial state Eve prepares with the $A$ and $T$ registers each consisting of $N$ copies of a $d$-dimensional Hilbert space.  From Theorem \ref{thm:sample}, using the sampling strategy discussed in Section \ref{sec:sample}, we know that there exist ideal states $\{\ket{\phi^t}\}_t$ such that $\ket{\phi^t} \in \text{span}(\mathcal{G}_{t,\delta})\otimes\mathcal{H}_E$ and, furthermore:
\[
\trd{\frac{1}{T}\sum_t\kb{t}\otimes\kb{\psi} - \frac{1}{T}\sum_t\kb{t}\otimes\kb{\phi^t}} \le \sqrt{\epsilon_\delta^{cl}} = \epsilon.
\]
Above, we are defining the quantum ideal set $\text{span}(\mathcal{G}_{t,\delta})$ with respect to the computational basis $\{\ket{0}, \cdots, \ket{d-1}\}$.  We will first analyze the ideal state $\sigma = \frac{1}{T}\sum_t\kb{t}\otimes\kb{\phi^t}$ and later promote this analysis to the real state $\kb{\psi}$ actually produced by the adversary.  

Consider $\sigma$; choosing a subset for sampling is equivalent to measuring the first register causing the state to collapse to a particular ideal state $\kb{\phi^t}$.  At this point, Alice and Bob measure their systems, indexed by $t$ in the $Z$ basis resulting in outcome $q_A$ and $q_B$ respectively (these are $m$-character strings in some $d$-letter alphabet).  Since $\ket{\phi^t} \in \text{span}(\mathcal{G}_{t,\delta})\otimes\mathcal{H}_E$, it is clear that, conditioning on measurement outcome $q_A, q_B$, the state collapses to one of the form:
\begin{equation}\label{eq:ideal-collapse}
  \ket{\phi^t(q_A,q_B)} = \ket{\mu} = \sum_{a,b\in J(q_A,q_B)}\widetilde{\alpha}_{a,b}\ket{a,b}_{AB}\otimes\ket{\widetilde{E}_{a,b}}
\end{equation}
where
\[
J(q_A,q_B) = \{(a,b)\in\al_d^n\times\al_d^n \st \hd(a,b) \dc \hd(q_A,q_B)\}.
\]
We may rewrite the above in the equivalent form:
\begin{equation}
  \ket{\mu} = \sum_{a\in\al_d^n}\alpha_a\ket{a} \otimes\sum_{b\in \mathcal{J}(q_A,q_B : a)} \beta_{b|a}\ket{b}\ket{E_{a,b}},
\end{equation}
with:
\begin{equation}
  \mathcal{J}(q_A,q_B:a) = \{b \in \al_d^n \st \hd(a,b) \dc \hd(q_A,q_B)\}.
\end{equation}
Note that some of the $\alpha$'s and $\beta$'s appearing above may be zero.

From this post measurement state, Bob will encode his key choice.  He chooses a random key $K = k_1\cdots k_n \in \mathcal{A}_d^n$ and applies $\mathcal{U}_K$ to his portion of the quantum state $\ket{\mu}$.  Let $\ket{\mu^K} = I_A\otimes\mathcal{U}_K\otimes I_E\ket{\mu}$.  Then it is not difficult to see that:
\begin{equation}
  \ket{\mu^K} = \sum_{a\in\al_d^n}\alpha_a\ket{a} \otimes\sum_{b\in \mathcal{J}(q_A,q_B : a)} \beta_{b|a}\exp(2\pi i (b\cdot K)/d)\ket{b}\ket{E_{a,b}},
\end{equation}
where $b\cdot K = b_1k_1 +\cdots + b_nk_n$.  This entire process of $B$ choosing $K$ and encoding his choice onto $\ket{\mu}$ may be described by the density operator:
\[
\sum_{K \in \al_d^n}\frac{1}{d^n}\kb{K}\otimes\kb{\mu^K}_{ATE}.
\]
At this point, Bob sends the $T$ register to Alice, keeping, of course, the $K$ register private.  Eve intercepts this return signal and is allowed to probe all $n$ of the returning qubits simultaneously with an attack which also may depend on her initial ancilla (see Figure \ref{fig:diagram}).  After this probe, Eve must forward an $n$-qubit register to Alice for her to complete the protocol.  Eve's goal is to gain information on the $K$ register held privately by Bob.  In particular, we need a bound on the min entropy $\Hmin(K|E')$ of the quantum state following this probe and the forwarding of the $n$ qubits to Alice.  However, we will actually compute a bound on $\Hmin(K|TE)$ before this probe and before Eve forwards an $n$-qubit register to Alice; from the data processing inequality, it is clear that this will serve as a lower bound for our desired $\Hmin(K|E')$, namely $\Hmin(K|E') \ge \Hmin(K|TE)$.

Consider $\sigma'$ as defined above.  Tracing out the $A$ register yields:
\[
\sigma'_{KTE} = \sum_{a\in\al_d^n}|\alpha_a|^2\underbrace{\sum_{K}\frac{1}{d^n}\kb{K}\otimes P\left(\sum_{b\in \mathcal{J}(q_A,q_B:a)}\beta_{b|a}\exp(2\pi i (b\cdot K)/d)\ket{b}\ket{E_{a,b}}\right)}_{\sigma_{KTE}^{(a)}}.
\]
From Equation \ref{eq:mixed-entropy}, we have $\Hmin(K|TE)_{\sigma'} \ge \min_a\Hmin(K|TE)_{\sigma^{(a)}}$.  We claim that for any $a$, it holds that:
\begin{equation}\label{eq:entropy-ideal-state}
\Hmin(K|TE)_{\sigma^{(a)}} \ge n\left(\log_2d - \frac{H_d(\hd(q_A,q_B) + \delta)}{\log_d 2}\right)
\end{equation}

To give some intuition behind this claim, imagine the ideal case where there is absolutely no noise (thus $b = a$ always) and no finite sampling imprecision (i.e., $\delta = 0$).  In this case $|\mathcal{J}(q_A,q_B:a)| = 1$ which implies that $\sigma^a_{KTE}$ is of the form:
\begin{align*}
  \sigma^a_{KTE} &= \sum_K\frac{1}{d^n}\kb{K}\otimes\left|\exp(2\pi i (a\cdot K)/d)\right|^2\kb{a}\otimes\kb{E_{a,a}}\\
  &=\sum_K\frac{1}{d^n}\kb{K}\otimes\kb{a}\otimes\kb{E_{a,a,}}.
\end{align*}
From this, it is obvious that $\Hmin(K|TE) = n\log_2 d$ since the $TE$ system is independent of the $K$ register.  Of course, this ideal case is impossible to achieve (even without noise, it would still hold that $\delta > 0$ in a finite key setting).  However, if the noise is ``small enough'' then the set $\mathcal{J}(q_A,q_B:a)$ is also ``small'' and, as we now show, the entropy cannot decrease by too much.

To prove our claim, we will use a proof technique from \cite{renner2008security} used to bound the min entropy of a superposition state.  While the proof of our claim is nearly identical to that used in \cite{renner2008security}, the statement is different and does not immediately follow from that prior work and, so, is worth stating and proving here.

From $\sigma^{(a)}_{KTE}$, consider instead the mixed state:
\[
\chi_{KTE} = \sum_K\frac{1}{d^n}\kb{K}\otimes\sum_{b\in\mathcal{J}(q_A,q_B:a)}|\beta_{b|a}|^2\kb{b}\otimes\kb{E_{a,b}}.
\]
Clearly $\Hmin(K|TE)_\chi = n\log_2 d$ since the $K$ and $TE$ registers are separable.  Let $\mathcal{J} = \mathcal{J}(q_A,q_B:a)$; we claim that $|\mathcal{J}|\chi \ge \sigma^{(a)}$ from which it holds that $\Hmin(K|TE)_{\sigma^{(a)}} \ge \Hmin(K|TE)_\chi - \log_2 |\mathcal{J}|$.  Since it is clear that $|\mathcal{J}| \le d^{nH_d(\hd(q_A,q_B)+\delta)}$ for any $a$ (this follows from the well known bound on the volume of a Hamming ball), this will imply our desired claim in Equation \ref{eq:entropy-ideal-state}.

Consider an arbitrary vector $\ket{\zeta} \in \mathcal{H}_K\otimes\mathcal{H}_{TE}$.  We may write $\ket{\zeta} = \sum_k\gamma_k\ket{k}\otimes\ket{\psi_k}_{TE}$.  Then:
\begin{align*}
  \braket{\zeta|\chi|\zeta} = \sum_k\frac{|\gamma_k|^2}{d^n}\sum_{b\in\mathcal{J}}|\beta_b|^2|\braket{\psi_k|b,E_{a,b}}|^2
\end{align*}
and:
\begin{align*}
  \braket{\zeta|\sigma^{(a)}|\zeta} = \sum_k\frac{|\gamma_k|^2}{d^n}\sum_{b,b'\in\mathcal{J}}\beta_b\beta_{b'}^*\braket{\psi_k|b,E_{a,b}}\braket{b',E_{a,b'}|\psi_k} = \sum_k\frac{|\gamma_k|^2}{d^n}\left|\sum_{b\in\mathcal{J}}\beta_b\braket{\psi_k|b,E_{a,b}}\right|^2
\end{align*}
From the Cauchy-Schwarz inequality, it holds that
\[
\braket{\zeta|\chi|\zeta} \ge \sum_k\frac{|\gamma_k|^2}{d^n}\frac{1}{|\mathcal{J}|}\left|\sum_{b\in\mathcal{J}}\beta_b\braket{\psi_k|b,E_{a,b}}\right| = \frac{1}{|\mathcal{J}|}\braket{\zeta|\sigma^{(a)}|\zeta}.
\]
Since $\ket{\zeta}$ was arbitrary, it follows that $|\mathcal{J}|\chi \ge \sigma^{(a)}$ as desired.  Thus, it holds that
\[
\Hmin(K|TE)_{\sigma'}\ge\min_a\Hmin(K|TE)_{\sigma^{(a)}} \ge n\log_2 d - \max_a\log_2|\mathcal{J}(q_A,q_B:a)| \ge n\left(\log_2 d - \frac{H_d(\hd(q_A,q_B)+\delta)}{\log_d 2}\right).
\]

Of course, this was only the entropy contained in the ideal state produced by Theorem \ref{thm:sample}.  However, Lemma \ref{lemma:entropy} (using the choice of $t$ and observation $q_A,q_B$ as the random variable $X$ in that theorem) completes the analysis.  In particular, except with probability at most $2\epsilon^{1/3}$, it holds that the conditional quantum min entropy in the real state, following the protocol's execution, is:
\begin{equation}
  \Hmin^{4\epsilon + 2\epsilon^{1/3}}(K|E)_{\psi(t,q_A,q_B)} \ge n\left(\log_2 d - \frac{H_d(\hd(q_A,q_B)+\delta)}{\log_d 2}\right).
\end{equation}
where $\ket{\psi(t,q_A,q_B)}_{ABE}$ is the actual quantum state produced by the adversary initially conditioned on a sample subset of $t$ and an observation $q_A$ and $q_B$.

Now that we have a bound on the quantum min entropy, we may compute the final expected secret key length.  Let $\epsilon_{PA} = 9\epsilon + 4\epsilon^{1/3}$ be the desired distance from the ideal secret key as in Equation \ref{eq:PA}. Then, except with a failure probability of $\epsilon_{fail} = 2\epsilon^{1/3}$, it holds that the protocol may distill a secret key of size:
\begin{equation}\label{eq:final-keyrate-ideal}
  \ell = n\left(\log_2 d - \frac{H_d(\hd(q_A,q_B)+\delta)}{\log_d 2}\right) - \leak - 2\log\frac{1}{\epsilon}
\end{equation}

Note that, to estimate $\leak$, parties would sample their raw key error rate.  This is purely a classical problem at that point and $\leak$ is a quantity that may be observed directly.  Later, when we evaluate, we will simulate a reasonable value for $\leak$ based on the channel noise observed in the forward and reverse channel. Though since that is a purely classical process we do not elaborate further on how this is done here.  Interestingly, later, when we consider imperfect devices, the noise in both channels (forward and reverse) will be needed to bound the quantum min entropy and so will warrant further discussion later.

\subsection{Evaluation}\label{section:eval1}

We now evaluate this protocol and compare also to high-dimensional BB84 \cite{bechmann2000quantum}.  To make a fair comparison, we will compare to two parallel copies of BB84 as was done in \cite{two-way-qkd-security} (though, there, only qubit systems were discussed); namely, we will assume that Alice and Bob are able to run two parallel ``sessions'' of BB84 using the two-way quantum channel thus effectively doubling the key-rate in that case.  To perform a numerical comparison, we will model both channels (the ``forward'' channel connecting Alice to Bob and the ``reverse'' channel connecting Bob to Alice) as depolarization channels.  Note that this assumption is not required for our security proof which works for any channel (the users must simply determine $\hd(q_A,q_B)$ and $\leak$ and then evaluate Equation \ref{eq:final-keyrate-ideal}).  Instead, we use this assumption only to evaluate and also to compare with prior work (which often assumes depolarization channels when evaluating key-rates).  A depolarization channel acting on a $d$-dimensional system $\rho$ acts as follows: $\mathcal{E}:\rho \mapsto (1-Q)\rho + Q/dI$, where $I$ is the $d$-dimensional identity operator.

We will actually consider two specific cases for the two-way quantum channel typically considered when evaluating two-way QKD protocols, namely dependent and independent channels.  In the independent case, each channel is modeled as two independent depolarization channels.  For the dependent case, the noise in the reverse channel may depend on the forward channel.  This may arise, for instance, in certain fiber implementations, where sending a photon back along the same channel it traveled originally will ``undue'' some of the noise picked up in the first channel \cite{lucamarini2014quantum,two-way-qkd-security}.

For the independent case, we may easily compute the expected error in the measurement string $q_A$ and $q_B$.  If the initial state prepared were $\ket{\psi_0}$ then, following the first depolarization channel, it will evolve to:
\[
\kb{\psi_0} \mapsto (1-Q)\kb{\psi_0} + \frac{Q}{d^2}I = \rho,
\]
from which we compute $\hd(q_A,q_B) = \frac{Q}{d^2}(d^2-d) = Q\left(1-\frac{1}{d}\right)$.  The remaining, unmeasured portions, have the key encoded into them and returned to Alice.  This return operation is modeled as a second depolarization channel.   If $\rho^k$ is the result of Bob encoding key-digit $k$ onto $\rho$, then the state after this second depolarization channel is easily found to be in the independent case:
\[
\rho^k_{ind} = (1-Q)^2\kb{\psi^k} + \frac{1-(1-Q)^2}{d^2}I
\]
where $\ket{\psi^k} = I_A\otimes U_k\ket{\psi_0}$.  In the \emph{dependent} case, we take the return state to be:
\[
\rho^k_{dep} = (1-Q)\kb{\psi^k} + \frac{Q}{d^2}I.
\]
Note that, in the independent case, the noise in the reverse channel is independent of the noise in the forward channel; whereas in the dependent case, this is not true.  Such scenarios can arise, as mentioned, in various fiber implementations.

Alice will then measure the two systems using POVM $\Lambda$ leading to her guess of Bob's encoding.  Any errors here will result in leakage due to error correction which we account for by setting $\leak = 1.2 H(A|B)$ and where $H(A|B)$ may be easily computed using the computed state $\rho^k$ above.  In particular, let $Pr(A=i | B=k)_{ind}$ be the conditional probability that, assuming Bob encoded key digit $k$, that Alice decodes a key digit of $i$) (in particular her $\Lambda_i$ element clicked) in the independent channel case.  From $\rho^k_{ind}$, this is easily found to be:
\begin{equation}
  Pr(A=i | B=k)_{ind} = \left\{\begin{array}{ll}
  (1-Q)^2 + \frac{1 - (1-Q)^2}{d} & \text{ if $k=i$}\\
  \frac{1-(1-Q)^2}{d} & \text{ if $k \ne i$}
  \end{array}\right.
\end{equation}
A similar expression may be found for the dependent case:
\begin{equation}
  Pr(A=i | B=k)_{dep} = \left\{\begin{array}{ll}
  (1-Q) + \frac{Q}{d} & \text{ if $k=i$}\\
  \frac{Q}{d} & \text{ if $k \ne i$}
  \end{array}\right.
\end{equation}

Of course $Pr(B = k) = 1/d$ (in both dependent and independent cases).  Consider, first, the independent case.  Let $x = (1-Q)^2 + \frac{1 - (1-Q)^2}{d}$ and $y = \frac{1-(1-Q)^2}{d}$.  Of course $x + (d-1)y = 1$. Simple algebra then shows that for the independent case we have:
\begin{align*}
  H(A|B) = H(AB) - H(B) &= -\frac{1}{d}\sum_k\left(x\log\frac{1}{d}x + \sum_{i\ne k} y\log\frac{1}{d}y\right) - \log d\\
  &=(x + (d-1)y)\log d - x\log x - (d-1)y\log y - \log d\\
  &= -x\log x - (d-1)y \log y.
\end{align*}
An identical expression, though changing $x = 1-Q + Q/d$ and $y = Q/d$ may be found for the dependent case.  This allows us to compute $\leak$ and thus evaluate $\ell$ in Equation \ref{eq:final-keyrate-ideal}.

The key-rate of the protocol, for dimension $d = 2, 4$, and $8$ is shown in Figure \ref{fig:keyrate1} for the dependent and independent cases.  We set $\epsilon = 10^{-36}$ thus giving both a failure probability, and an $\epsilon'$-secret key, on the order of $10^{-12}$. We note that, as the dimension increases, the key rate also increases.   Interestingly, this is not simply due to the fact that, as the dimension increases, the number of possible raw key digits from one signal state increases.  If this were the case, running multiple parallel copies of lower dimensions would yield the same result.  However, this is clearly not the case as shown in Figure \ref{fig:parallel-comp} where we compare two parallel executions of the protocol running with $d=2$ to a single execution of the protocol running with $d=4$.  We also note that, for higher dimensions, the number of signals needed to first achieve a positive key-rate decreases as the dimension increases.  This work thus shows even more advantages to high-dimensional quantum states that cannot be achieved by simple, naive, parallel executions of a lower-dimensional protocol.

In Figure \ref{fig:bb84-comp}, we compare the key-rate of this two-way protocol to running two independent copies of (high-dimensional) BB84.  For evaluating the high-dimensional BB84 key-rate we use results from \cite{HD-BB84-keyrate}.  Interestingly, the key-rate of BB84 in the qubit-case is exactly twice that of the two-way protocol (thus the advantage is only in the parallel executions of BB84 and not in the protocol itself in this ideal device scenario), yet for higher dimensions, the key-rate of BB84 is more than twice the rate of the two-way protocol analyzed here (implying BB84 as a protocol has an advantage over this two-way protocol in higher dimensions).  Finally, we evaluate the key-rate of the protocol as a function of the raw key error when the number of signals is large ($N = 10^{20}$) and compare to BB84 in Figure \ref{fig:keyrate-noise}.  Here, the raw-key error is defined to be the probability that Alice and Bob's key digits disagree.  We note that the noise tolerance when $d=2$ and dependent channels is identical to that of standard qubit BB84.  These similarities to BB84 were discovered also in \cite{two-way-qkd-security} for the qubit case but for a more complicated two-way protocol involving four encoding operations (at the qubit level).  Rather surprisingly, our work here shows that these additional encoding operations do not help the noise tolerance of this protocol in the qubit case.  However, also interestingly, this similarity between BB84 and the two-way protocol here no longer holds for $d > 2$.  In higher dimensions, BB84 always outperforms the two-way protocol.

\begin{figure}
    \centering
    \includegraphics[width=.5\textwidth]{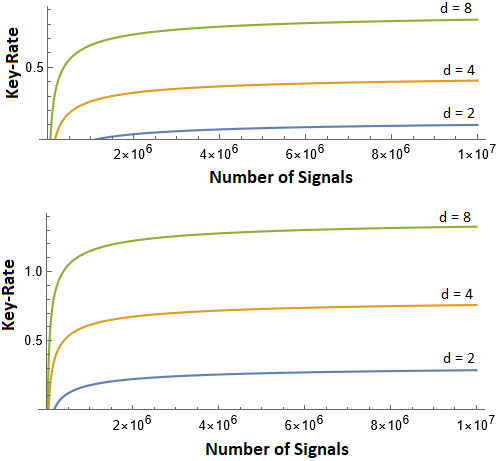}
    \caption{Evaluating the key-rate of our protocol under a depolarization channel.  Top: Assuming the forward and reverse channels are independent (thus ultimately creating additional noise in the reverse channel); Bottom: Assuming the forward and reverse channel noise is dependent.  We note that as the dimension increases, the key-rate also increases and the number of signals needed before a positive key rate can be attained decreases.  In both graphs, we take $Q = 10\%$.}
    \label{fig:keyrate1}
\end{figure}

\begin{figure}
    \centering
    \includegraphics[width=.5\textwidth]{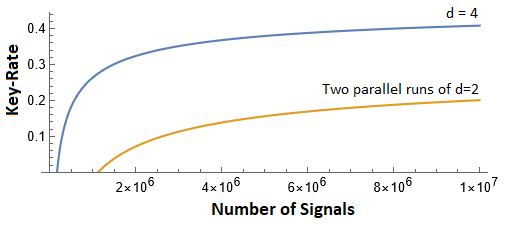}
    \caption{Comparing the key-rate of this protocol setting $d=4$ to two independent copies of running the protocol set at $d=2$.  Note that two copies of a two dimensional protocol cannot outperform the four dimensional version in both key-rate and also the point at which the key-rate becomes positive.  Here, we have $Q = 10\%$.}
    \label{fig:parallel-comp}
\end{figure}

\begin{figure}
    \centering
    \includegraphics[width=.5\textwidth]{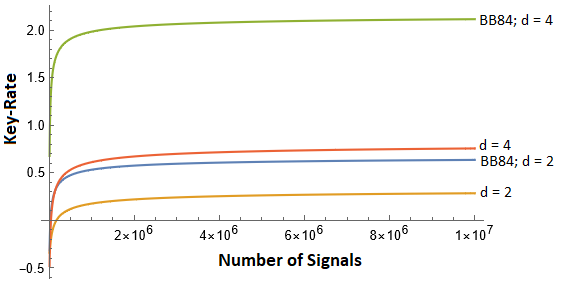}
    \caption{Comparing the key-rate of this protocol under a dependent channel assumption with two independent copies of BB84. Here, we have $d=4$ and $Q = 10\%$.  We note that, for the $d=2$ case, the key-rate of \emph{one copy} of BB84 converges to the key-rate of this two-way protocol in the dependent case.  For higher dimensions, however, BB84 outperforms the two-way protocol.  The graph of the independent channel case is similar, but with, as expected due to its increased noise levels, a more drastic difference between the two protocols.}
    \label{fig:bb84-comp}
\end{figure}

\begin{figure}
    \centering
    \includegraphics[width=.5\textwidth]{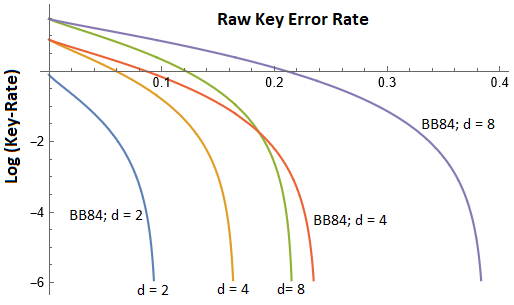}
    \caption{Evaluating the noise tolerance of the protocol in terms of raw key error rates (namely the probability that Alice and Bob's key digits disagree) in the dependent case.  Also comparing to HD-BB84.  We note that, when $d=2$ the noise tolerance of both protocols agree and tends towards $11\%$ (it is slightly lower here, due to our imperfect error correction factor of $1.2H(A|B)$; if the $1.2$ is replaced with $1$, both protocols achieve a noise tolerance of $11\%$).  For $d > 2$, BB84 outperforms.}
    \label{fig:keyrate-noise}
\end{figure}

\section{Imperfect Measurements}

We now consider the case where Alice and Bob's measurement devices are not ideal.  In particular, their measurements may have inefficiencies leading to missed detections.  Furthermore, we add the possibility of vacuum events (e.g., photon loss).  In this case, the protocol is identical as before, except that, now, if Alice detects a vacuum when attempting to decode Bob's key choice (either a real vacuum or a detector miss due to a low efficiency), she will signal to Bob to discard that particular round's key digit.  The detection of vacuums by Bob and Alice in the check stage will be recorded and used in estimating Eve's uncertainty.  Our contribution in this section is to derive a finite key expression for this two-way protocol under such conditions and to show how our proof method, applied in the previous section for the ideal device scenario, can be readily adapted to more complicated circumstances such as this.  Finally, our method does not require characterization of the device efficiency, thus leading to a (very) partial form of device independence on the measurement devices.  This is in contrast to our work in \cite{krawec-sample-qkd} for HD-BB84 where only vacuum pulses were analyzed but detectors were perfect and fully characterized; here we tackle the more difficult problem where devices are also inefficient.

Now, Bob's measurement may be described by the POVM $\mathcal{Z} = \{Z_0,\cdots, Z_{d-1}, Z_{vac}\}$ where:
\begin{align}
  &Z_j = \eta\kb{j}\\
  &Z_{vac} = I - \sum_j Z_j,\notag
\end{align}
where $\eta \in [0,1]$ represents the detector's efficiency (with $\eta=1$ meaning perfect devices).  We also increase the dimension of the underlying Hilbert space to $d+1$ and use $\ket{vac}$ to represent the vacuum state which satisfies $\braket{vac|j} = 0$ for all $j$.  Note that $Z_{vac}$ will click if Eve sends a vacuum pulse (or if there is loss), however it may also click if one of the detectors which ideally would have clicked (if $\eta = 1$) missed the signal.  Importantly, users cannot tell the difference between a vacuum state and a missed detection.  A similar POVM may be defined for Alice (who also needs one for the key decoding measurement).  Interestingly, our analysis does not require actual knowledge of $\eta$.

To analyze this scenario, we return to the ideal states $\{\ket{\phi^t}\}$ which now live in a larger Hilbert space to accommodate the additional vacuum state.  Let $\ket{\psi}_{ATE}$ be the state produced by the source (Eve in our case), where the Alice and Bob systems consist of $N$ tensor copies of a $d+1$ dimensional Hilbert space spanned by $\{\ket{0}, \cdots, \ket{d-1}, \ket{vac}\}$.  As we are using the same classical sampling strategy as before, by Theorem \ref{thm:sample}, there exist ideal states such that:
\begin{equation}\label{eq:idealstate-2}
  \ket{\phi^t} \in \text{span}\left(\ket{a,b} \st \hd(a_t,b_t) \dc \hd(a_{-t}, b_{-t})\right)\otimes\mathcal{H}_E.
\end{equation}
where the average over all subsets is $\epsilon$-close in trace distance to the real state (we use the same $\delta$ as in Equation \ref{eq:delta}).  Note that, above, the strings $a$ and $b$ are in a $d+1$ character alphabet in order to accommodate the additional vacuum state.

As before, we analyze the ideal state first, though our analysis is more involved in this case as we will need to consider the second attack operator this time.  Note that, previously, we could ignore the second attack operator to Eve's advantage (we computed Eve's uncertainty before this second probe and before sending additional qubits to Alice - her second probe and transmission would only \emph{increase} her uncertainty and so by ignoring this, we got a worst-case lower-bound on the key rate).  Now, however, the second probe can introduce loss and, furthermore, Alice will discard all rounds where she did not detect a signal.  Thus, the final raw key will be conditional, based on the observations of vacuums and, so, we can no longer ignore the second attack in this instance.  Eve's second attack can directly influence the final key beyond simply creating additional errors that are taken into account using the $\leak$ term.  See Figure \ref{fig:diagram}.

Consider a particular run of the protocol on the ideal state $\sigma = \sum_t\frac{1}{T}\kb{t}\otimes\kb{\phi^t}$.  Let $t$ be the chosen subset and $q_A, q_B \in \{0,1,\cdots,d-1,vac\}^m \cong \al_{d+1}^m$ be the observed outcome for Alice and Bob using POVM $\mathcal{Z}$.  Note that, unlike before, we are not making an ideal basis measurement (in the basis $\{\ket{0}, \cdots, \ket{d-1}, \ket{vac}\}$)  and, so, observing a particular $q_A$ and $q_B$ does not immediately translate to direct information on $a_t$ and $b_t$ in Equation \ref{eq:idealstate-2}.  That is, unlike our earlier analysis where parties knew that $a_t = q_A$ and $b_t = q_B$ due to them making an ideal basis measurement, the value of $a_t$ and $b_t$ cannot be directly observed; instead we must use our observation to obtain suitable constraints on what they might be.  Thus, our challenge now is to use this observation to determine a ``worst case'' condition on the structure of the collapsed ideal state.  In particular, whenever a ``vacuum'' is observed, Alice or Bob cannot be certain if it is due to an actual vacuum or a detector inefficiency.  Thus, they must assume the worst that it could have been any of the $d+1$ symbols.  Therefore, vacuum observations lead to a significant amount of uncertainty in the final collapsed state of the ideal superposition.

More formally, consider the post measured state after measuring $\ket{\phi^t}$ and observing this $q_A, q_B$.  Due to the structure of $\ket{\phi^t}$ and $\mathcal{Z}$, it is not difficult to see that the post measurement state of the unmeasured portion must be of the form:
\begin{equation}
  \sigma(t,q_A,q_B) = \sum_{\widetilde{q}_A, \widetilde{q}_B\in V(q_A,q_B)}p(\widetilde{q}_A,\widetilde{q}_B)P\left(\sum_{a,b \in J(\widetilde{q}_A,\widetilde{q}_B)}\hat{\alpha}_{a,b,\widetilde{q}_A,\widetilde{q}_B}\ket{a,b}\ket{\hat{E}_{a,b,\widetilde{q}_A,\widetilde{q}_B}}\right),
\end{equation}
where $\sum p(\widetilde{q}_A,\widetilde{q}_B) = 1$, $p(x,y) \ge 0$, and:
\begin{align}
  V(q_A,q_B) = \left\{(x,y) \in \al_{d+1}^n\times\al_{d+1}^n \st \begin{array}{c}x^\ell = q_A^\ell \text{ if } q_A^\ell \ne vac\\y^\ell = q_B^\ell \text{ if } q_B^\ell \ne vac\end{array}\right\}\label{eq:Vset}\\
  J(\widetilde{q}_A,\widetilde{q}_B) = \{(a,b)\in \al_{d+1}^n\times \al_{d+1}^n \st \hd(\widetilde{q}_A,\widetilde{q}_B) \dc \hd(a,b)\}.
\end{align}
Above, $x^\ell$ represents the $\ell$'th character of word $x$ (similar for the other character strings).  Here, $V(q_A, q_B)$ represents the uncertainty that Alice and Bob have on the post-measured ideal state due to device imperfections.  Namely, whenever the $\ell$'th signal results in a vacuum, as discussed, they cannot be certain whether it is really a vacuum or one of the other $\ket{x^\ell}$.  Whenever the $\ell$'th signal results in a non-vacuum, however, they can be certain the state collapsed to $\kb{x^\ell}$.  This set $V$ is the set of all strings that agree with the observed $q_A$ and $q_B$ given the definition of our measurement POVM - whenever the measurement produces a non-vacuum state, users are certain of the underlying symbol measured; otherwise there is uncertainty and the symbol might have been any of the $d+1$ characters.  From this uncertainty, then, the set $J$ denotes the potential basis states that are $\delta$-close to what might have been observed if devices were ideal.  Note that, if $\eta = 1$ and if there is no vacuum state, it holds that $V(q_A,q_B) = \{(q_A,q_B)\}$ and so the state above is equivalent to the ideal case in Equation \ref{eq:ideal-collapse} as expected.

At this point, Bob encodes his key choice, leading to the state:
\begin{align}
   &\frac{1}{d^n}\sum_{k\in\al_d^n}\kb{k}\otimes\left(\sum_{\widetilde{q}_A,\widetilde{q}_B}p(\widetilde{q}_A,\widetilde{q}_B)P\left[\sum_{a,b}\hat{\alpha}_{a,b,\widetilde{q}_A,\widetilde{q}_B}\exp(2\pi i (b\cdot k)/d)\ket{a,b}_{AT}\ket{\hat{E}_{a,b,\widetilde{q}_A,\widetilde{q}_B}}\right]\right)\notag\\
  =&\frac{1}{d^n}\sum_k\kb{k}\otimes\left(\sum_{\widetilde{q}_A,\widetilde{q}_B}p(\widetilde{q}_A,\widetilde{q}_B)P\left[\sum_{a\in\al_{d+1}^n}\alpha_{a,\widetilde{q}_A,\widetilde{q}_B}\ket{a}\sum_{b\in J(\widetilde{q}_A,\widetilde{q}_B:a)}\beta_{a,b,\widetilde{q}_A,\widetilde{q}_B}\exp(2\pi i (b \cdot k)/d)\ket{b}_T\ket{E_{a,b,\widetilde{q}_A,\widetilde{q}_B}}\right]\right).\label{eq:state-enc-imperfect}
\end{align}
where:
\begin{equation}
  J(\widetilde{q}_A,\widetilde{q}_B:a) = \{b \in \al_{d+1}^n \st \hd(a,b) \dc \hd(\widetilde{q}_A,\widetilde{q}_B)\}.
\end{equation}
Note that, for the key encoding, we assume that it leaves the vacuum state untouched.  In particular, for $x \in \{0, \cdots, d-1\}$ and $y \in \{0, \cdots, d-1\}$, we have:
\[
U_y\ket{x} = \exp(2\pi i xy/d)\ket{x}
\]
while:
\[
U_y\ket{vac} = \ket{vac}.
\]
Thus, in Equation \ref{eq:state-enc-imperfect}, by $b\cdot k$, for some $b \in \al_{d+1}^n$, we mean $\tilde{b}\cdot k$ where $\tilde{b}\in\al_d^n$ is the same as word $b$ but with every vacuum symbol replaced with a zero.

At this point, the $T$ system returns to Eve's control.  As mentioned earlier, however, it is not sufficient to compute Eve's uncertainty holding the transit register at this point.  This is due to the fact that Eve's reverse attack is allowed to create vacuum events which will later be discarded by Alice.  Thus, in a way, Eve has some control over which key bits in the $K$ register are to be used and this must be taken into account.  Before, in the ideal device case considered earlier, all key bits were used regardless of Alice's measurement - the only effect Eve had at that point would be to increase the error rate, thus causing additional information leakage due to error correction (which is taken into account in the $\leak$ term of our key-rate evaluation).  Here, $E$'s second attack can potentially increase her information by discarding rounds strategically.

We will model Eve's reverse attack as a unitary operation, $U_R$, acting on all $n$ qudits and her previous ancilla state simultaneously.  Without loss of generality, the action of this operation may be written:
\begin{equation}
U_R\ket{b, E_{a,b,\widetilde{q}_A,\widetilde{q}_B}} = \sum_{c\in\mathcal{A}_{d+1}^n}\ket{c}\ket{F_{a,b,c,\widetilde{q}_A,\widetilde{q}_B}}
\end{equation}
where we assume the $d+1$'th character in $\mathcal{A}_{d+1}$ represents the vacuum state.  After applying this attack operation, the state becomes:
\[
\tau = \frac{1}{K}\sum_k\kb{k}\otimes\sum_{\wideq}p(\wideq)P\left(\sum_{a,c\in\al_{d+1}^n}\alpha_{a,\wideq}\ket{a,c}_{AT}\sum_{b\in J(\wideq :a)}\beta_{a,b,\wideq}\exp(2\pi i (b\cdot k)/d) \ket{F_{a,b,c,\wideq}}\right).
\]

Now, the $T$ register is returned to Alice who performs a measurement of the joint $AT$ register using the POVM $\Lambda = \{\Lambda_0,\cdots, \Lambda_{d-1},\Lambda_{vac}\}$ where:
\[
\Lambda_k = \eta\sum_a P\left(\sum_j\exp(2\pi i k j /d)\ket{j,j+a}\right)
\]
and:
\[
\Lambda_{vac} = I - \sum_k\Lambda_k.
\]
Note that, for Alice, a vacuum outcome results from either a missed detection (based on $\eta$) or if either or both of the $A$ or $T$ registers contain a vacuum for that signal state.  After this measurement, Alice receives some outcome $\bar{a} \in \{0, 1, \cdots, d-1, vac\}$ and reports all indices where she received a vacuum state (as those will be discarded).  Let ${m}(\bar{a}) = m_1\cdots m_n$ where $m_i$ is 1 only if the $i$'th index of $\bar{a}$ is a vacuum ($m_i = 0$ otherwise).  That is ${m}(\bar{a})$ is the classical bit string message that Alice sends to Bob marking all rounds where she received a vacuum.  Since this is a message sent on the public channel, we record it in a separate register.  For any particular $m(\bar{a})$, Bob will discard the corresponding bits of his key register.  Let $v(\bar{a})$ be the number of vacuum events in $\bar{a}$.  Thus, the state now becomes:
\begin{equation}
\sum_{\bar{a}}p_A(\bar{a})\kb{\bar{a}}\otimes\tau^{t,q,\bar{a}}_{KE},
\end{equation}
where the state $\tau^{t,q,\bar{a}}_{KE}$ is found to be (after tracing out both the $A$ and $T$ registers, along with the discarded key digit systems):
\begin{align*}
  &\frac{1}{d^{n-v(\bar{a})}}\sum_{k_s\in\mathcal{A}_d^{n-v(\bar{a})}}\kb{k_s}\otimes\\
  &\sum_{k_r\in\mathcal{A}^{v(\bar{a})}}\frac{1}{d^{v(\bar{a})}}\sum_{a,c,\wideq}|\alpha_{a,\wideq}|^2P\left(\sum_{b\in J(\wideq : a)}\beta_{a,b,\wideq}\exp(2\pi ib\cdot f(k_s,k_r,m(\bar{a}))/d)\ket{F_{a,b,c\wideq}}\right).
\end{align*}
where, above, $f$ is a function that correctly reconstructs $k$ using $k_s$ and $k_r$ by placing the ``keep'' key digits ($k_s$) in the correct location relative to the ``reject'' digits ($k_r$); clearly this depends only on $m(\bar{a})$.

At this point, we bound the min entropy conditioning on a particular $t$, $\bar{a}$, $q_A$, and $q_B$ outcome.  We use a similar analysis method as in Section \ref{sec:security1} to bound the min entropy, namely by defining a similar mixed state, computing the entropy for that mixed state, and showing that the entropy of the superposition state is similar (the algebra from the ideal-device case analyzed earlier follows exactly in this non-ideal device case).  In particular, we find:
\begin{align*}
\Hmin(K|E) &\ge (n-v(\bar{a}))\log_2 d - \max_{a,c,k_r,\wideq}\log_2|J(\wideq : a)|\\
&\ge (n-v(\bar{a}))\log_2 d -  \max_{\wideq \in V(q_A,q_B)}\frac{n\Hextd_{d+1}(\hd(\wideq) + \delta)}{\log_{d+1} 2},
\end{align*}
where the last inequality can be shown using the same analysis method as earlier on bounding the size of the set $J(\wideq:a)$.

We therefore need only to find the maximum Hamming distance of all $\wideq \in V(q_A,q_B)$, given the observed $q_A$ and $q_B$.  However, it is not difficult to see by the definition of $V(q_A,q_B)$ (Equation \ref{eq:Vset}) that this is maximized whenever we set index $j$ of $\widetilde{q}_A$ and $\widetilde{q}_B$ to be different symbols every time either $q_A$ or $q_B$ is a vacuum in index $j$.  Said differently, whenever a symbol in ${q}_A$ is a non-vacuum, the corresponding symbol in $A$'s element of an entry in $V(q_A, q_B)$ must be that symbol; otherwise if it is a vacuum, the symbol can be anything.  Similarly for Bob.  This gives an upper-bound on the size of the superposition set based only on observed parameters, thus giving a worst-case bound on the min entropy.  The min-entropy can only be higher in reality (thus, the key-rate can only be better than what we derive here).  Let $x,y\in \{0,1,\cdots, d-1, vac\}^n$ and define:
\begin{equation}
\widetilde{\Delta}_H(x,y) = \frac{1}{n}\sum_j g(x_j, y_j),
\end{equation}
where:
\begin{equation}
  g(x_j, y_j) = \left\{\begin{array}{rl}
  1 & \text{ if } x \ne y \text{, or } x = vac \text{, or } y = vac\\
  0 & \text{ otherwise}\end{array}\right.
\end{equation}
Clearly this function $\widetilde{\Delta}_H(x,y)$ is counting the number of errors in the $x$ and $y$ strings while also treating any vacuum symbol in either $x$ or $y$ as an error.  We can then conclude that:
\begin{equation}
\Hmin(K|E) \ge (n-v(\bar{a}))\log_2 d -  \frac{n\Hextd_{d+1}(\widetilde{\Delta}_H(q_A,q_B) + \delta)}{\log_{d+1} 2}
\end{equation}

Of course, this is only the ideal state analysis.  However, by using Lemma \ref{lemma:entropy}, we may promote this to the real-state analysis.  Let $\epsilon_{PA} = 9\epsilon + 4\epsilon^{1/3}$, then, except with a failure probability of $\epsilon_{fail} = 2\epsilon^{1/3}$, the final secret key size is:
\begin{equation}
  \ell = (n-v(\bar{a}))\log_2 d -  \frac{n\Hextd_{d+1}(\widetilde{\Delta}_H(q_A,q_B) + \delta)}{\log_{d+1} 2} - \leak - 2\log\frac{1}{\epsilon}.
\end{equation}

Note that the above analysis did not take into account the actual value of $\eta$; instead any vacuum event was analyzed in a way that gave full advantage to the adversary.  If $\eta$ is actually known, tighter bounds may be derived potentially.  One may use a more complex sampling strategy which takes into account the number of vacuum counts, instead of treating all vacuum as errors.  However, our analysis above does not require Alice and Bob to know the actual value of $\eta$ thus giving a (very partial) form of device independence in the measurement devices, though at the potential cost of decreased key-rate as parties must assume the worst always.

\subsection{Evaluation}

We evaluate as in the ideal case considered earlier, assuming a depolarization channel (both dependent and independent).  To compute the expected observation $\widetilde{\Delta}_H(q_A,q_B)$, we have:
\[
\widetilde{\Delta}(q_A, q_B) = \frac{1}{m} \left(| \{ \text{vacuum symbols in $q_A$ or $q_B$}\}| + |\{\text{number of errors in remaining symbols}\}|\right).
\]
We assume for evaluation purposes the probability of a vacuum event is $\mu$ (this can be due to loss, detector inefficiencies, or both).  Then, assuming a depolarization channel with parameter $Q$ (see Section \ref{section:eval1}), we have that:
\begin{equation}
  \widetilde{\Delta}_H(q_A, q_B) = \mu + (1-\mu)Q\left(1 - \frac{1}{d}\right).
\end{equation}

We also require an expected value for $\nu(\bar{a})$, the number of vacuum events in Alice's second measurement.  For this, we evaluate the simple case where a loss can occur either in the first channel, or, if not in the first channel, then in the second.  Thus:
\[
\frac{1}{n}\nu(\bar{a}) = \mu + (1-\mu)\mu.
\]
Of course, this is only for our evaluation purposes; our security proof works for any observed value of $\nu(\bar{a})$.
Finally, for $\leak$, we use the same analysis as in Section \ref{section:eval1} to compute the error correction leakage.  This gives us everything needed to evaluate our bound.

We evaluate the key-rate of the protocol for different values of $\mu$ (the probability of a vacuum click) in Figure \ref{fig:keyrate-loss}.  We observe that this protocol is highly sensitive to vacuum events, though as the dimension increases, the protocol can withstand this better.  Whether this is a consequence of the protocol, or our security proof in this section not being tight, is an open question.   Though we do suspect it is more a consequence of the protocol being highly sensitive to noise.  Indeed, it has already been discovered that the qubit Ping Pong protocol is highly sensitive to loss by devising specific attacks \cite{wojcik2003eavesdropping,zhang2004improving,bostrom2008security} (with the attack in \cite{zhang2004improving} breaking security with only $25\%$ loss).  Our bound shows a loss tolerance much lower than this $25\%$, however we are also analyzing general attacks - our security analysis covers prior attacks but also new, undiscovered ones.  Thus, it may in fact be the case that this protocol is this highly sensitive to noise in the finite key setting and other countermeasures may be necessary.  None the less, we still have demonstrated a positive key rate is possible under general attacks with loss and that the dimension of the underlying signal can benefit key-rates under lossy channels.  

\begin{figure}
    \centering
    \includegraphics[width=.45\textwidth]{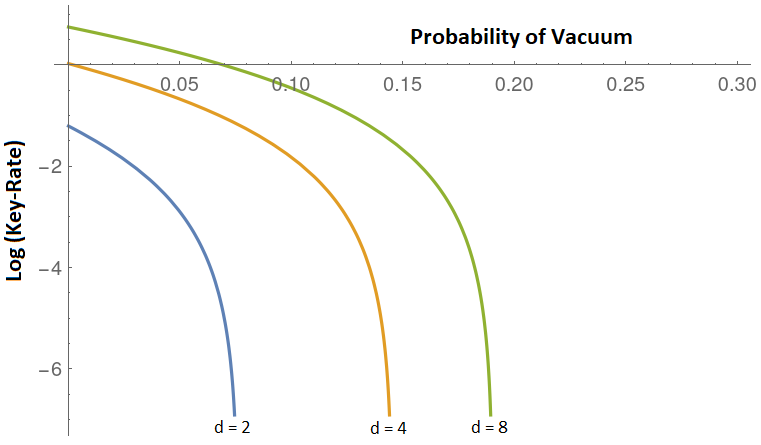}
    \includegraphics[width=.45\textwidth]{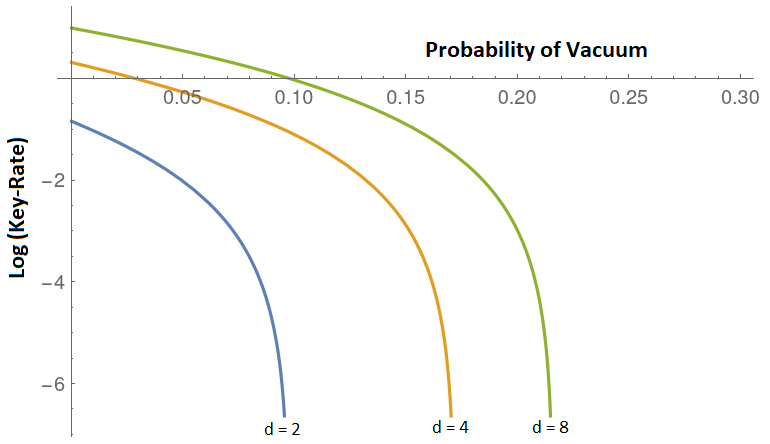}
    \caption{Evaluating the key-rate of the two-way protocol assuming loss and imperfect detectors. Here $Q = 5\%$ and the number of signals is set to $10^{20}$.  Left: Assuming independent depolarization channels; Right: Dependent channels.  See text for discussion.}
    \label{fig:keyrate-loss}
\end{figure}

\section{Closing Remarks}


In this paper, we considered the high-dimensional Ping Pong protocol introduced in \cite{QKD-HD-PP}.  In particular, we considered a variant of the protocol placing fewer device requirements on users (thus making it potentially easier to implement in practice while also creating some interesting theoretical problems).  Importantly, we derived an information theoretic security proof in the finite key setting.  Our methods may be broadly applicable to other QKD protocols which are not immediately reducible to full-entanglement based versions (making standard tools difficult, or impossible, to use).  For example, our proof approach may be applied to two-way protocols that do not exhibit the necessary symmetry requirements needed to use results in \cite{two-way-qkd-security} for reducing to entanglement-based versions, thus creating opportunities to analyze protocols where standard tools can't be immediately applied.  We also showed how our methods can be used to analyze imperfect detectors and lossy channels.  Finally, we performed a rigorous evaluation of the protocol's performance in a variety of dimensions and scenarios.  Our work provides additional evidence, beyond that already known as mentioned in the introduction, that high-dimensional quantum states may benefit, at least in theory, quantum communication, including in the two-way case.

Several interesting problems remain open.  Our proof does not require an exact characterization of the efficiency of the detectors, leading to a pessimistic key-rate bound under loss.  Whether this is a consequence of our proof method in this scenario, or that the protocol itself has low loss tolerance (due in part to its two-way channel), is an open question.  Also, adapting our proof to work with stronger forms of device independence could be interesting.  One potential way forward there may be to consider measurement devices that are not characterized, except for their overlap (as was done in \cite{two-way-qkd-security} for two-way qubit based protocols obeying certain symmetry properties).

\section*{Acknowledgments}
The author would like to acknowledge support from NSF grant number 2006126.


\end{document}